\documentclass[12pt]{article}

\usepackage{amsmath}

\usepackage{amsfonts,amssymb,amsthm}
\usepackage{mathtools}
\usepackage{nicefrac,setspace}
\usepackage{xcolor,authblk}

\usepackage{geometry}
\usepackage{hyperref}       
\usepackage{url}            
\usepackage{booktabs}       
\usepackage{amsfonts}       
\usepackage{nicefrac}       
\usepackage{microtype}      

\newtheorem{theorem}{Theorem}

\newtheorem*{theorem*}{Theorem}

\newtheorem*{claim*}{Claim}
\newtheorem*{remark*}{Remark}

\newtheorem*{lemma*}{Lemma}

\newtheorem{lemma}[theorem]{Lemma}
\newtheorem{corollary}[theorem]{Corollary}
\newtheorem*{corollary*}{Corollary}

\newtheorem{prop}[theorem]{Proposition}
\newtheorem{remark}[theorem]{Remark}
\newtheorem{problem}{Problem}

\newcommand{\conv}{{\mathsf{conv}}}

\newcommand{\sep}{{\mathsf{sep}}}

\newcommand{\R}{\mathbb{R}}

\newcommand{\xc}{\mathrm{xc}}

\newcommand{\bits}{\{0,1\}}

\DeclareMathOperator{\sgn}{sgn}

\begin{document}

\title{On the extension complexity of polytopes separating subsets of the Boolean cube}

\author{Pavel Hrube{\v s}\footnote{\tt pahrubes@gmail.com} \qquad Navid Talebanfard\footnote{\tt talebanfard@math.cas.cz}}
\affil{Institute of Mathematics \\ Czech Academy of Sciences}

\date{}

\maketitle

\begin{abstract}
We show that 
\begin{enumerate} 
\item for every $A\subseteq \bits^n$, there exists a polytope $P\subseteq \R^n$ with $P \cap \bits^n = A$ and extension complexity $O(2^{n/2})$, 
\item there exists an $A\subseteq \bits^n$ such that the extension complexity of any $P$ with $P\cap \bits^n = A$ must be at least $2^{\frac{n}{3}(1-o(1))}$.   
\end{enumerate}
We also remark that the extension complexity of any 0/1-polytope in $\R^n$ is at most $O(2^n/n)$ and pose the problem whether the upper bound can be improved to $O(2^{cn})$, for $c<1$. 
\end{abstract}

\section{Introduction}
A polytope $P\subseteq \R^n$ with many facets can often be expressed as a projection of a higher-dimensional polytope $Q\subseteq \R^m$ with much fewer facets. This is especially significant in the context of linear programming: instead of optimizing a linear function over $P$, it is more efficient to optimize over $Q$. \emph{Extension complexity of} $P$ 
is defined as the smallest $k$ so that $P$ is an affine image of a polytope with $k$ facets. Extension complexity has been studied in the seminal paper of Yannakakis \cite{Yannanakis}, Fiorini et al., Rothvo{\ss}  \cite{FioriniExp, Rothvoss}, and others. In \cite{Rothvoss}, Rothvo{\ss} has shown that there exist 0/1-polytopes in $\R^n$ with extension complexity $2^{\frac{n}{2}(1-o(1))}$; in fact, a random polytope has this property. 
Our paper originated with the question whether the bound of Rothvo{\ss} is tight. 
\begin{problem}\label{problem1}
Can every 0/1-polytope $P$ be expressed as a projection of a polytope with $O(2^{cn})$ facets, for a constant $c<1$? 
\end{problem}
Note that $P$ itself can have many more than $2^n$ facets \cite{Barany}. Extension complexity, however, can be bounded by the number of vertices and hence the trivial upper bound is $2^n$. In Problem \ref{problem1}, we want to know which of the bounds, $2^{n/2}$ vs. $2^n$, is closer to the truth. This is reminiscent of a similar problem in $\R^2$. In \cite{FioriniPolygons}, Fiorini et al. have shown that there exist polygons in $\R^2$ with $k$ vertices and  extension complexity $\Omega(\sqrt{k})$. Quite surprisingly, Shitov \cite{Shitov} has shown that every $k$-vertex polygon has extension complexity $O(k^{2/3})$.
Furthermore, there is an $O(k^{1/2})$
upper bound for several natural classes of polygons \cite{Kwan}.  

Problem \ref{problem1} is related to a similar question about graphs. Given an $n$-vertex graph, let $P_G\subseteq \R^n$ be the convex hull of characteristic vectors of its edges. An explicit description of $P_G$ in terms of inequalities is known \cite{EdgeHull,EdgeHull2}, and it is especially simple in the case of bipartite graphs. The trivial upper bound on $\xc(P_G)$ is $O(n^2)$. We point out that any improvement on this trivial bound gives an improvement on extension complexity of 0/1-polytopes. This reduction is similar to the so-called \emph{graph complexity} (see \cite{DBLP:journals/acta/PudlakRS88, MR3307088}) where an $n$ variate Boolean function is interpreted as defining a graph on exponentially many vertices. Extension complexity of $P_G$ has been investigated by Fiorini et al.  in \cite{FioriniStable}, where a non-trivial upper bound $O(n^2/\log n)$ was presented (cf. \cite{StableSet}). This yields a modest contribution to Problem \ref{problem1}: $P$ is a projection of a polytope with $O(2^n/n)$ facets. 

We mainly focus on a relaxation of Problem \ref{problem1}. Given $A\subseteq \{0,1\}^n$, a polytope $P\subseteq \R^n$ will be called \emph{separating} for $A$ if $P\cap \bits^n= A$. 
In other words, $P$ separates Boolean points in $A$ from those outside of $A$. The smallest separating polytope is $\conv(A)$ itself. Extension complexity of separating polytopes has several connections with computational complexity, as extensively discussed in \cite{ja:epsilon}. Here we show that every $A\subseteq \bits^n$ has a separating polytope with extension complexity $O(2^{n/2})$. This is achieved by the aforementioned reduction to graphs, and by showing that the set of edges of $G$ has a separating polytope of linear size.  This quantitatively matches the lower bound of Rothvo{\ss}  -- except that the assumptions are different. There are infinitely many separating polytopes other than $\conv(A)$ itself and so the lower bound is not applicable. In \cite{ja:real}, a lower bound of $2^{\Omega(n)}$ on extension complexity of separating polytopes has been given. The constant in the exponent hinges on known bounds on quantifier elimination and it is not hard to see that the proof from \cite{ja:real} gives $2^{\frac{n}{5}(1-o(1))}$. We will improve this bound to $2^{\frac{n}{3}(1-o(1))}$ using a more geometrical argument.    

\section{The main result}

A \emph{polytope} $P\subseteq \R^n$ is the convex hull of a finite set of points in $\R^n$. 
It can also be viewed as a bounded set defined by a finite number of linear constraints.  The \emph{extension complexity} of a polytope  $P$, $\xc(P)$, is the smallest $k$ so that there exists a polytope $Q\subseteq \R^m$ with $k$ facets and an affine map $\pi:\R^m\rightarrow \R^n$  such that $P=\pi(Q)$. Given $A\subseteq \bits^n$, its \emph{separation complexity}, $\sep(A)$, is the minimum $\xc(P)$ over all polytopes $P\subseteq \R^n$ with 
\[P\cap \bits^n = A\,;\]
such a $P$ is called a \emph{separating} polytope for $A$.

We provide non-trivial upper and lower bounds on $\sep(A)$:

\begin{theorem}\label{thm:main} \begin{enumerate}\item For every $A\subseteq \bits^n$, $\sep(A)\leq O(2^{n/2})$.
\item There exists $A\subseteq \bits^n$ with $\sep(A)\geq 2^{\frac{n}{3}(1-o(1))}$. 
\end{enumerate}
\end{theorem}

\begin{remark} In \cite{ja:epsilon,ja:real}, separation complexity is defined slightly differently with $P$ allowed to be an \emph{unbounded} polyhedron. This is just a cosmetic detail -- we can intersect $P$ with $[0,1]^n$ (or a simplex containing it) which increases its complexity by an additive term of $O(n)$. 
\end{remark}

As observed in \cite{Yannanakis,ja:epsilon}, a Boolean circuit of size $s$ which accepts precisely the inputs from $A$ gives a separating polytope for $A$ with extension complexity $O(s+n)$. 
This means that an upper bound of $O(2^n/n)$ on $\sep(A)$ can be obtained from known upper bounds on Boolean circuits  due to Lupanov (see Section 1.4.1 in \cite{MR2895965}). 
A lower bound of $2^{\Omega(n)}$ on separation complexity has been obtained in \cite{ja:real}. 
Both bounds presented in Theorem \ref{thm:main} are quantitatively stronger.

\section{Simple bounds on the number of facets}

We first give some elementary bounds on the number of facets of separating polytopes. This is mainly to contrast it both with extension complexity and the number of facets of $0/1$-polytopes. 

\begin{prop}
\label{prop:facets1} 
Every $A\subseteq \bits^n$ has a separating polytope with at most $2^n$ facets. 
\end{prop} 

\begin{proof} For $x\in \R^n$ and $\sigma\in \bits^n$, define
\[ h_\sigma(x):= \sum_{i=1}^n x_i(1-\sigma_i)+(1-x_i)\sigma_i\,.\]
If $x$ is Boolean, $h_\sigma(x)$ is the Hamming distance between $x$ and $\sigma$. 
Define $P\subseteq \R^n$ by the constraints $h_\sigma(x)\geq 1$ for every $\sigma\in \bits^n\setminus A$. Then 
indeed $P\cap \bits^n= A$. $P$ may be possibly unbounded. This can be remedied by adding the constraints $h_\sigma(x)\geq 0$, $\sigma\in A$.
\end{proof}

Let $\mathsf{ODD}_n\subseteq \bits^n$ be the set of Boolean strings with odd number of ones. 

\begin{prop}\label{prop:facets2} If $n\geq 2$, every separating polytope for $\mathsf{ODD}_n$ has at least $2^{n-1}$ facets. 
\end{prop}

\begin{proof} Let $H$ be a closed half-space $\{x\in \R^n: \sum_i a_ix_i\geq b\}$. We claim the following: if $\mathsf{ODD}_n\subseteq H$ then $\bar{H}:= \R^n\setminus H$ contains at most one even string $\sigma\in \bits^n\setminus \mathsf{ODD}_n$. To see this, assume an even $\sigma$ is contained in $\bar{H}$. Without loss of generality, we can assume that $\sigma$ is the zero vector; otherwise apply an affine map that flips $0$ and $1$ for the $1$-coordinates of $\sigma$. Since $\sigma\not \in {H}$, we must have $b>0$.
Since every unit vector is in $\mathsf{ODD}_n \subseteq {H}$, we have $a_i\geq b$ for every $i$. This means that $\bits^n\setminus\{\sigma\}\subseteq H$ and no other even string can be in $\bar{H}$. 

If $n= 2$, the statement of the proposition is clear. Let $n\geq 3$ and assume that $P$ is a separating polytope for $\mathsf{ODD}_n$ with $r$ facets. Then $P = \bigcap_{i=1}^r H_i$ where $H_i$ are closed half-spaces. (This is because $P$ is full-dimensional for $n \ge 3$). We have $\mathsf{ODD}_n\subseteq H_i$ for every $i$, and every even $\sigma$ is contained in at least one $\bar{H}_i$. Since $|\bits^n\setminus \mathsf{ODD}_n|=2^{n-1}$, this gives $r\geq 2^{n-1}$.  
\end{proof}

By the result of B\'{a}r\'{a}ny \cite{Barany}, $A$ can have $2^{\Omega(n\log n)}$ facets -- hence Proposition \ref{prop:facets1} shows that a separating polytope for $A$ can have much fewer facets than $\conv(A)$. The convex hull of $\mathsf{ODD}_n$, also known as the parity polytope, has extension complexity $O(n)$ (see \cite{Carr}), and it is trivially a separating polytope for $\mathsf{ODD}_n$ -- hence Proposition \ref{prop:facets2} shows that taking extensions into a higher dimension can be exponentially powerful. It also shows that the $O(2^{n/2})$ upper bound from Theorem \ref{thm:main} cannot be achieved simply by counting the facets of the separating polytope.

\section{The upper bound}

We now prove the upper bound from Theorem \ref{thm:main}.

Let $B_2^n\subseteq \{0,1\}^n$ be the set of Boolean vectors of Hamming weight two (i.e., with exactly two ones).  For a natural number $n$, $[n]$ will denote the set $\{1,\dots, n\}$.

\begin{lemma}\label{lem:graph} Let $H\subseteq B^n_2$. Then there exists a polytope $R_H\subseteq [0,1]^n\subseteq \R^n$ with at most $2n$ facets such that 
$R_H\cap \{0,1\}^n= H$. 
\end{lemma}

\begin{proof} 
It is convenient to view $H$ as representing edges of a graph with vertex set $[n]$. Namely, $i\not=j$ are adjacent iff $e_i+e_j\in H$,  where $e_i$ is the $i$-th unit vector. 
Given $i\in [n]$, let $N(i)$ be the set of vertices adjacent with $i$.    

Let $R_H$ be defined by the following constraints
\begin{align} 
0\leq x_i\,,\, i\in [n] & \,, \nonumber \\
\sum_{i\in [n]} x_i =2 & \,,\label{g2} \\
x_i\leq \sum_{j\in N(i)} x_j,\, i\in [n]\,. &\label{g3}  
\end{align}
There are $2n$ inequalities. It is easy to see they imply $x_i\leq 1$ for every $i\in [n]$ and so $R_H\subseteq [0,1]^n$. Given $e_i+e_j\in H$, the constraints defining $R_H$ are satisfied and so $H\subseteq R_H$. If $\sigma\in \{0,1\}^n\setminus R_H$ then either $\sigma\not \in B^n_2$, and then $\sigma$ falsifies (\ref{g2}), or $\sigma= e_i+e_j$ with $j\not\in N(i)$, and then $\sigma$ falsifies $x_i\leq \sum_{k\in N(i)} x_k$. Hence  $R_H\cap \{0,1\}^n= H$.   
\end{proof}

\begin{theorem}\label{thm:UB} Let $A\subseteq \{0,1\}^n$. Then there exists a polytope $P\subseteq [0,1]^n$ with $\xc(P)=O(2^{n/2})$ and $P\cap \{0,1\}^n= A$. 
\end{theorem}

\begin{proof} Without loss of generality, assume that $n$ is even and $N:= 2^{n/2}$. Assume $A\subseteq \{0,1\}^{[n]}$ and partition $[n]$ into two equal parts $X_1$ and $X_2$. 
Let $F_1:= \{0,1\}^{X_1}$ and $F_2:= \{0,1\}^{X_2}$. 
Hence every $\sigma\in \{0,1\}^{[n]}$ can be uniquely written as $\sigma=\sigma_1\cup \sigma_2$ with $\sigma_1\in F_1$, $\sigma_2\in F_2$.  
We identify $\R^{2N}$ with $\R^{F_1\cup F_2}$, so that the coordinates are indexed by elements of $F_1\cup F_2$. The standard unit vectors are $e_{\sigma_1}, e_{\sigma_2}$, $\sigma_1\in F_1$, $\sigma_2\in F_2$. Let $H\subseteq B^{2N}_2$ be defined as
\[ H:= \{e_{\sigma_1}+e_{\sigma_2}:\, \sigma_1\cup \sigma_2\in A \}\,.\]
Let $R_H$ be the polytope from the previous lemma. We want to express $P$ in terms of $R_H$.

Let $T\subseteq \R^{2N}$ be the intersection of $[0,1]^{2N}$ with the hyperplanes
\begin{equation} \sum_{\sigma_1\in F_1} x_{\sigma_1}=1\,,\,\, \sum_{\sigma_2\in F_2} x_{\sigma_2}=1 \,.\label{g4}\end{equation}
Let $\pi: \R^{2N}\rightarrow \R^n$ be the linear map so that for every $\sigma_1\in F_1$, $\sigma_2\in F_2$,
$\pi(e_{\sigma_1})=\sigma_1\cup 0$ and $\pi(e_{\sigma_2})=0\cup \sigma_2$ (where $0$ is the zero vector in $F_2$ and $F_1$, respectively).
This guarantees 
 \[\pi (e_{\sigma_1}+e_{\sigma_2})= \sigma_1\cup \sigma_2\,.\]  
Moreover, for every $\sigma\in \{0,1\}^n$ with $\sigma=\sigma_1\cup \sigma_2$, $e_{\sigma_1}+e_{\sigma_2}$ is the \emph{unique} vector in $x\in T$ with $\pi(x)=\sigma$. For if $\pi(x)=\sigma$, we have $\sum_{\beta\in F_1} x_\beta = \sigma_1$ and $\sum_{\beta \in F_1} x_\beta=1$ by (\ref{g4}). (Similarly for $\sigma_2$). In other words, $x$ gives a convex combination of $\sigma_1$ in terms of the Boolean vectors $F_1$ which is easily seen to be unique. 

Finally, let $P:= \pi(R_H \cap T)$. Then $\xc(P)\leq 2N$. Given $\sigma\in \{0,1\}^n$, we have $\sigma\in P$ iff $e_{\sigma_1}+e_{\sigma_2}\in R_H$. By the definition of $H$, this is
equivalent to $\sigma\in A$. Hence indeed  $P\cap \{0,1\}^n= A$.
\end{proof}

\subsection{Graphs and Problem \ref{problem1}}

Given a (simple undirected) graph $G$ with vertex set $[n]$, let $P_G\subseteq \R^n$ be the convex hull of characteristic vectors of edges in $G$:
\[P_G=\conv(\{e_i+e_j: \hbox{$i\not =j$ are adjacent in } G\}).\] 
Note that $P_G$ has at most ${n}\choose {2}$ vertices and so $\xc(P_G)$ is at most quadratic. Fiorini et al.  in \cite{FioriniStable} have given an improved bound $\xc(P_G)\leq O(n^2/\log n)$ for any 
graph.  It is however not known whether $\xc(P_G)\leq O(n^c)$ for some constant $c<2$. 
We summarize the connection between this problem and Problem \ref{problem1} as follows:

\begin{prop} Let $A\subseteq \{0,1\}^n$. Then $\xc(\conv(A))\leq O(2^n/n)$. Moreover, assume that for every bipartite graph $G$ on  $2m$ vertices (with the parts of equal size), $\xc(P_G)\leq O(m^c)$, where $c\leq 2$ is an absolute constant. Then $\xc(A)\leq O(2^{cn/2})$   
\end{prop}  

\begin{proof} This is analogous to the proof of Theorem \ref{thm:UB}. The set $H$ corresponds to a bipartite graph $G$ on $2N$ vertices with $N=2^{n/2}$. The projection $\pi$ maps vertices of $P_G\subseteq T$ to Boolean vectors in $\R^n$.  Hence $\pi(P_G)=\conv(A)$ and $\xc(\conv(A))\leq \xc(P_G)$. From \cite{FioriniStable}, Lemma 3.4, we know $\xc(P_G)\leq O(N^2/\log N)$ which gives  $\xc(\conv(A))\leq O(2^n/n)$; the ``moreover'' part is similar. 
\end{proof}

An explicit description of $P_G$ in terms of linear inequalities can be found in \cite{EdgeHull,EdgeHull2}. Apart from the general constraints $\sum x_i=2$, $0\leq x_i$, every inequality 
$\sum_{i\in S}x_i\leq 1$ is valid whenever $S$ is an independent set. In the case of \emph{bipartite} $G$, this indeed gives a complete description of $P_G$. This also means that $P_G$ can have exponentially many facets -- in particular, the polytope from Lemma \ref{lem:graph} must be strictly larger than $P_G$ for some $G$.  

The lemma can be somewhat strengthened when considering independent sets of size $2$. Let $Q_G$ be the polyhedron defined by the constraints  $x_i+x_j\leq 1$ for every $i\not=j$ not adjacent in $G$. Clearly, $P_G\subseteq Q_G$ and they contain the same set of Boolean vectors of Hamming weight two (i.e., the edges of $G$). 

\begin{remark} \label{remark1} Let $G$ be a bipartite $n$-vertex graph. Then there exists a polytope $R'_G$ with $O(n)$ facets with $P_G\subseteq R'_G\subseteq Q_G$.
\end{remark}  

\begin{proof} Let $L$ and $R$ be the parts of $G$ with $L\cup R=[n]$. Then $R'_G$ defined by the following constraints has the desired properties:
\begin{align*} 
0\leq x_i\,,\, i\in [n]\,, & \\
\sum_{i\in L} x_i =1\,, \sum_{i\in R} x_i =1\,,  & \\
x_i+ \sum_{j\in R\setminus N(i)}x_j\leq 1,\, i\in L\,. &  
\end{align*}
 
\end{proof}

Using the machinery of non-negative rank factorizations of slack matrices (see, e.g., \cite{Yannanakis,Rothvoss,FioriniExp}), the quantity $\xc(P_G)$ can be captured by the nonnegative rank of an explicit matrix $\mathsf{EIS}_G$:
its rows are indexed by edges $e$ of $G$, columns by independent sets $S$. The entry corresponding to $e$ and $S$ equals $1$, if $e$ and $S$ are disjoint, and $0$ otherwise. 
This matrix is intimately related to the famous Clique vs Independent set problem of Yannakakis \cite{Yannanakis}; see also \cite{GoosCIS}. If $G$ is bipartite, $\xc(P_G)$ corresponds to the non-negative rank of $\mathsf{EIS}_G$. 
An interesting submatrix of  $\mathsf{EIS}_G$ is the $\mathsf{ENE}_G$ matrix obtained by restricting the columns to independent sets of size two (i.e., non-edges). A similar matrix has been considered in \cite{KushilevitzENE} from the point of view of communication complexity. Since $\mathsf{ENE}_G$ can have size $O(n^2)\times O(n^2)$, one may perhaps hope to obtain quadratic lower bounds on $\xc(P_G)$ using the non-negative rank of $\mathsf{ENE}_G$. We note that this is impossible\footnote{This could also be concluded from Remark \ref{remark1}}. 

\begin{remark}\label{remark2} If $G$ is a bipartite $n$-vertex graph, $\mathsf{ENE}_G$ can be written as a sum of $O(n)$ 0/1-matrices. For a non-bipartite $G$, the bound is $O(n\log n)$. 
\end{remark}   

\begin{proof}  Let $G$ be a bipartite graph on vertices $L\cup R$. Let $E$ be the set of edges of $G$ and $\bar{E}$ the set of non-edges.  
Given $\ell\in L$, $r\in R$, define the following sets $A_\ell, B_r, C_\ell \subseteq E\times \bar E$ of edge/non-edge pairs. 
\begin{enumerate}\item
$A_\ell$ consists of pairs with $e=\{\ell_1,r_1\}$, $\bar e= \{\ell, r_2\}$ with $\ell\not=\ell_1\in L$, $r_1,r_2\in R$ and $r_1\in N(\ell)$.   
\item
$B_r$ consists of  pairs $e= \{\ell_1, r\}$, $\bar e= \{v_1,v_2\}$, where either $v_1\not= v_2 \in R\setminus\{r\}$, or $v_1\in L$, $v_2\in R\setminus \{r\}$ and $v_1\not \in N(r)$. 
\item
$C_\ell$ are the pairs $\{\ell,r_1\}$, $\{\ell_1,\ell_2\}$ with $\ell_1\not =\ell_2\in L\setminus \{\ell\}$.  
\end{enumerate}
It is easy to see that the sets form a partition of the set of disjoint edge/non-edge pairs. Moreover, each of the sets is a product set (of the form $C\times C'$ with $C\subseteq E, C'\subseteq \bar E$). Identifying a subset of $E\times \bar E$ with the 0/1-matrix representing its characteristic function, we can thus write \[\mathsf{ENE}_G=\sum_\ell A_\ell+\sum_r B_r+\sum_\ell C_\ell\,,\]  
where the summands are 0/1-matrices of rank one.

A general $n$-vertex graph can be expressed as an edge-disjoint union of a bipartite graph and two graphs with $\lceil n/2\rceil$ vertices, and we can proceed by induction.  
\end{proof}

\subsection{The lower bound}

Our proof of the lower bound from Theorem \ref{thm:main} uses Warren's estimate on the number of sign patterns of a polynomial map, and Alon's bound on the number of combinatorial types of polytopes. We overview these results first.  

For $b \in \R $ define $$\sgn(b) = \begin{cases}  1, & b > 0, \\ 0, & b = 0 \\ -1, & b<0. \end{cases}$$ For a sequence $f = \langle f_1(y_1, \ldots, y_p), \ldots, f_m(y_1, \ldots, y_p)\rangle$ of real functions and $b \in \R ^p$, let $\sgn(f(b)) := \langle \sgn(f_1(b)), \ldots, \sgn(f_m(b))\rangle \in \{-1, 0, +1\}^m$, which we call the {\it sign-pattern of $f$ at $b$}. A result of Warren and its extension by Alon gives a bound on the number of sign patterns when $f_i$ are  polynomials of degree at most $d$.

\begin{theorem}[Warren \cite{warren}, Alon \cite{AlonTools}]
\label{thm:warren}
Let $f$ be a sequence of $m$ polynomials of degree at most $d\geq 1$ in the same set of $p$ variables with $2m\geq p$. Then $|\{\sgn(f(b)) : b \in \R^p\}| \le (\frac{8edm}{p})^p$.
\end{theorem}


Given a polytope $P$, the {\it face lattice} of $P$, $L(P)$, is the poset of the faces of $P$ ordered by inclusion (including $\emptyset$ and $P$ itself).  
It is naturally equipped with join and meet operations, hence it is a lattice.  
See, e.g., \cite{MR1311028} for details.
The lattice-isomorphism equivalence class of  $L(P)$ captures the {\it combinatorial type} of $P$.

\begin{theorem}[Alon \cite{MR859498}]
The number of non-isomorphic face lattices arising from polytopes with $r$ vertices is at most $2^{r^3(1 + o(1))}$.
\end{theorem}

By duality, this implies:

\begin{corollary}
\label{cor:type}
The number of non-isomorphic face lattices arising from polytopes with $r$ facets is at most $2^{r^3(1 + o(1))}$.
\end{corollary}

We now proceed to prove the lower bound from Theorem \ref{thm:main}. 
We call a set $S\subseteq \R^n$ \emph{full-dimensional} if no hyperplane in $\R^n$ contains $S$. Note that if $A$ is full-dimensional then so is any separating $P$ for $A$. 

\begin{lemma}
\label{lm:full-dim}
There are at least $2^{2^n(1 - o(1))}$ full-dimensional subsets of $\bits^n$.
\end{lemma}

\begin{proof}
If $A$ contains $0$ and the $n$ unit vectors, it is full-dimensional. There are $2^{2^n-n-1}$ such $A$'s. 
\end{proof}


\begin{lemma}
\label{lm:func-count}

For every $m\geq n$ there are polynomials $f_1, \ldots, f_s$  
in $mn$ variables such that \begin{enumerate} 
\item $s = O(m^{n+1})$, each $f_i$ has degree at most $n$ and has at most $2^{O(n\log n)}$ non-zero coefficients, 
\item for every set $V \in \R^{m\times n}$ viewed as $m$ points in $\R^n$, if $\conv(V)$ is full-dimensional then the set $\conv(V) \cap \{0, 1\}^n$ is uniquely determined by $$\langle\sgn(f_1(V)), \ldots, \sgn(f_s(V))\rangle.$$ 
\end{enumerate}

\end{lemma}

\begin{proof}

We will construct a set of polynomials such that for any $V = \{ v_1, \ldots, v_m\}$, we can determine $\conv(V) \cap \{0, 1\}^n$ by evaluating the signs of these polynomials on $V$. The idea is as follows. For every set $V'$  of $n$ points from $V$, we can compute the unique hyperplane $H$ passing through $V'$ (if one exists). If all points in $V$ lie in the same closed half-space determined by $H$, then $\conv(V) \cap H$ is a facet of $\conv(V)$. Let us call such closed half-space {\it good}. Then,  given $\sigma \in \{0, 1\}^n$, we can determine whether $\sigma\in \conv(V)$ by checking whether it appears in all good half-spaces. 

We now formally define our set of polynomials. 
Given $S \in \binom{[m]}{n}$ and $V\in \R^{m\times n}$, let $V_S$ be the set of vectors $\{v_i: i\in S\}$.  
We start by constructing the following polynomials/sets of polynomials. They take $V_S$ as input, but we hide the dependence.     

\begin{enumerate}
\item\label{a1} $a_{S,1},\dots a_{S,n}$ are polynomials of degree $n-1$  such that $V_S$ is affinely independent iff some $a_{S,i}$ is non-zero.
\item\label{a2}  
 $b_{S}$ is a polynomial  of degree $n$ such that whenever $V_S$ is affinely independent then  
$H_{S}(V):=\{x\in \R^n: \sum_i a_{S,i} x_i = b_{S}\}$ is the unique hyperplane passing through $V_S$, 
\item\label{a3} $F_S$ is a set of $m-n$ polynomials of degree $n$ such that 
if $V_S$ is affinely independent, then $\conv(V) \cap H_S(V)$ is a facet of $\conv(V)$ iff all polynomials in $F_S$ are all non-positive or all non-negative.
\end{enumerate}

Parts \ref{a1}, \ref{a2} are an exercise in linear algebra. 
$F_S$ is obtained by evaluating the hyperplane equation from \ref{a2} on all points from $V\setminus V_S$ -- the hyperplane defines a facet if all points in $V$ lie on the same side. 
 
Let $F$ the set of polynomials containing $F_S$, $a_{S,1},\dots, a_{S,n}$ and 
\begin{equation*}\sum_i a_{S,i}\sigma_i - b_{S}\,,\label{eq:side}\end{equation*}
for every $S\in \binom{[m]}{n}$ and $\sigma\in \bits^n$. Then $\conv(V)\cap \bits^n$ is uniquely determined by the signs of polynomials in $F$. 
The number of polynomials is $(2^n+n+1+(m-n))\binom{m}{n}\leq (2^n+m+1) \frac{m^n}{n!}\leq O(m^{n+1})$ and their degrees are at most $n$. The bound on the number of non-zero coefficients follows by noting that each polynomial depends on $O(n^2)$ variables. 
\end{proof}

Before proceeding to the next lemma, let us make some comments about rational functions. Given an $n$-variate rational function $f=g/h$ with $g,h$ coprime polynomials and $h\not=0$, define its degree as the maximum of the degrees of $g$ and $h$. Thus $f$ defines a partial function $:\R^n\rightarrow \R$. Warren's estimate can be extended to rational functions as follows. Given a rational map $f=\langle f_1,\dots, f_m \rangle$ from $\R^p$ to $\R^m$ with each $f_i$ of degree at most $d\geq 1$, we have 
\begin{equation}\label{eq:Warren2} 
|\{\sgn(f(b)) : b \in \R^p\,,\, f(b) \hbox{ is defined}\}| \le (cdm)^p\,,\end{equation}
where $c>0$ is an absolute constant. This follows from Theorem \ref{thm:warren} by  considering signs of numerators and denominators separately.\footnote{We also have no assumption on $p$ since the number of sign patterns can be trivially bounded by $3^p$.} 

Furthermore, we need the following estimate on the degree of composition. Suppose that $f(x_1,\dots,x_m)$ is a polynomial of degree $d_1$ with $k$ non-zero coefficients and $g_1,\dots, g_m$ are rational functions of degree at most $d_2$. Then it is easy to see that the degree of $f(g_1,\dots,g_m)$ is at most $kd_1d_2$. 

\begin{lemma}
\label{lm:type-count}
Let $L$ be a face lattice of a $d$-dimensional polytope with $r\geq d$ facets. Assume $d\geq n$ and let $S_L$ be the set of $A\subseteq \bits^n$ such that $A$ is full-dimensional and there exists a polytope $Q$ in $\R^d$ with combinatorial type $L$ such that the projection of $Q$ on the first $n$ coordinates is separating for $A$. Then $|S_L|\leq 2^{O(nr^3)}$.   
\end{lemma}

\begin{proof}
Consider a polytope $Q$ in $\R^d$ with face lattice $L$. Since $Q$ is full-dimensional, we can write it as $\{y\in \R^d: By \leq b\}$ where $b\in \R^r$, $B\in \R^{r\times d}$. Hence $Q$ can be described using $p:=(r+1)d \leq O(r^2)$ constants $z=\langle B,b\rangle$. Let $U$ be the vertices of $Q$. Then $|U|\leq 2^r$.    
Every vertex is the unique intersection of $d$ hyperplanes defining facets of the polytope. Furthermore, the lattice $L$ specifies for each vertex, which facets it is contained in and, moreover, which $d$ of them have the desired unique intersection (see, e.g., \cite{MR1311028}). For each vertex $u\in U$, canonically pick $d$ such facets. Then $u$ is the unique solution to a system of $d$ linear equations, and its coordinates can be seen as rational functions of $z$. More exactly, using Cramer's rule, we can write $u(z)= u_0(z)^{-1}\langle u_1(z),\dots, u_d(z) \rangle$, where $u_0,\dots,u_d$ have degree $d$. Note that $u_0(z)$ is non-zero whenever the polytope described by $z$ is indeed of type $L$. 

Project $Q$ on the first $n$ coordinates to obtain $P\subseteq \R^n$. We want to specify which elements of the Boolean cube are contained in $P$. Let $V$ be the projection of the vertices of $Q$ so that $P = \conv(V)$. Assume that $P$ is full-dimensional (otherwise it cannot contain a full-dimensional $A$). Lemma \ref{lm:func-count} gives us a set of polynomials $f_1(V), \ldots, f_s(V)$ whose sign pattern determines $P \cap \{0, 1\}^n$. We also have $s\leq O(|U|^{n+1})\leq 2^{O(rn)}$, and each $f_i$ has degree at most $n$ and $2^{O(n\log n)}$ non-zero coefficients.  The coordinates of vertices of $V$ are degree $d$ rational functions of $z$, hence $f_i(V(z))$ is  a rational functions of $z$ of  degree at most $d'\leq dn 2^{O(n\log n)}   \le r2^{O(n\log n)}$.  By  (\ref{eq:Warren2}),  the number of sign patterns of $\langle f_1(V(z)), \ldots, f_s(V(z))\rangle$ 
can be bounded by $(c's d')^{p}$. Since $s\leq 2^{O(rn)}$, $p\leq O(r^2)$, $d'\leq r2^{O(n\log n)}$, and $n\leq r$, the bound can be written as  $2^{O(r^3 n)}$.  This gives the desired estimate on $|S_L|$.
\end{proof}

\begin{theorem}
\label{thm:main-lb}
There exists $A\in \{0, 1\}^n$ such that $\sep(A) \ge 2^{\frac{n}{3}(1 - o(1))}$.
\end{theorem}

\begin{proof} Let $\cal A$ be the set of full-dimensional subsets $A\subseteq \bits^n$.  Let $r \ge n$ be such that every $A\in {\cal A}$ has separation complexity at most $r$. Without loss of generality,  assume  that this is exhibited by a full-dimensional polytope $Q \subseteq \R ^{d}$ with $r$ facets such that the projection of $Q$ on the first $n$ coordinates is a separating polytope for $A$, and $n\leq d\leq r$.
We then have  \[|{\cal A}|\leq |{\cal L}|\cdot \max_{L\in {\cal L}} |S_L|\,,\] where $\cal L$ is the set of combinatorial types of polytopes with $r$ facets.  
By Lemma \ref{lm:type-count},  $|S_L|\leq 2^{O(nr^3)}$. By Corollary \ref{cor:type}, we have  $|{\cal L}|\leq 2^{r^3(1 + o(1))}$. Therefore $|{\cal A}|\leq 2^{cnr^3}$ (for some constant $c$ and $n$ sufficiently large). By Lemma \ref{lm:full-dim}, we must have $2^{cnr^3} \ge 2^{2^n(1 - o(1))}$ and thus $r \ge 2^{\frac{n}{3}(1 - o(1))}$.
\end{proof}

\section{Acknowledgment}

We are grateful to Fedor Part for discussions. This work was supported by  GA{\v C}R grant 19-27871X and was done while Talebanfard was participating in the program {\it Satisfiability: Theory, Practice, and Beyond} at the Simons Institute for the Theory of Computing.

\bibliographystyle{plain}
\bibliography{myrefs,MYref,ref}

\end{document}